\documentclass[12pt,letterpaper]{article}
\usepackage{amssymb}
\usepackage{amsfonts}
\usepackage{amsmath}
\usepackage{amsthm}
\usepackage[nohead]{geometry}
\usepackage{setspace}
\usepackage{indentfirst}
\usepackage{endnotes}
\usepackage{graphicx}
\usepackage{rotating}
\usepackage{geometry}
\usepackage{tikz}
\usepackage{enumerate}
\usepackage{nicefrac}
\usepackage{subfig}
\usepackage{multirow}
\usepackage[authoryear,round]{natbib}
\usetikzlibrary{arrows,decorations.text,decorations.pathmorphing,backgrounds,fit,calc}

\usepackage{pgfplots}
\pgfplotsset{compat=newest}
\pgfplotsset{plot coordinates/math parser=false}
\newlength\figureheight
\newlength\figurewidth

\newtheorem{theorem}{Theorem}

\newtheorem{lemma}[theorem]{Lemma}

\newtheorem{proposition}[theorem]{Proposition}
\newtheorem{remark}[theorem]{Remark}

\DeclareMathOperator{\sign}{sgn}

\begin{document}

\title{Consumer Privacy and Serial Monopoly\thanks{Thanks to Konrad Mierendorff, Roland Strausz and seminar audiences at Bonn, ESSET Gerzensee, Humboldt University, EUI, EIEF, University of Bristol, Royal Holloway, QMUL and UCLA for helpful comments. Bhaskar thanks  the National Science Foundation, for its support via  grant \# SES 1629055.}}
\author{V. Bhaskar \\
%EndAName
UT Austin \and Nikita Roketskiy \\
%EndAName
UCL}
\maketitle
\thispagestyle{empty}
\begin{abstract}
We examine the implications of consumer privacy when  preferences today depend upon past consumption choices, and consumers shop from different sellers in each period. Although consumers are ex ante identical, their initial consumption choices cannot be deterministic. Thus ex post heterogeneity in preferences arises endogenously. 
Consumer privacy improves social welfare, consumer surplus and the profits of the second-period seller, while reducing the profits of  the first period seller,  relative to the situation where consumption choices are observed by the later seller. 


\end{abstract}
\newpage

\setcounter{page}{1}

\section{Introduction}

Do supermarkets encourage
wasteful purchases of goods with a limited shelf-life, such as fruit or
vegetables? Do fast-food restaurants induce over-consumption?
An important feature of consumers' preferences in these markets is
inter-temporal substitutability: if a patron has a heavy lunch, she is less hungry at dinner. If a consumer buys a two-for-one deal on salads or ready meals at the supermarket, he is less likely to buy similar goods when stopping at the local store. In these instances, purchases at different dates are substitutes. A different set of examples concerns habit formation or goods where taste can gradually develop. A student who frequents music concerts is more likely to enjoy them later in life, so that consumptions across dates are complements. Here, the concern is that the  club or venue today may not take into account the effect of its decisions upon  clubs in other cities tomorrow. 
In both types of example, purchases at
different dates are often from different suppliers. In the context of grocery stores and supermarkets, multi-store shopping is a widespread and well-documented phenomenon: a large share of consumers visit more than one store, often on different days.\footnote{  \cite{fox2004consumer} and \cite{thomassen2017multi} document evidence of multi-stop shopping in the US and  in the UK respectively.}

In this paper, we study optimal nonlinear pricing in the presence of either intertemporal substitutability or complementarity. Consumers' privacy plays an important role in our analysis. It has a profound effect on equilibrium per-unit price dispersion, social welfare and the allocation of surplus between the market participants. In particular, consumers are better off and social welfare is higher when transactions are private. 

We develop a model in which consumer's willingness to pay for a good depends on past  consumption, and in which a consumer who shops with a supplier today is unlikely to return tomorrow. To focus on the dynamic implications of endogenous choices, we assume that consumers are identical--- differences in taste 
arise only  due to differences in past consumption. Furthermore, we assume that a supplier at any date has monopoly power---e.g., because of search frictions---so that the market is characterized by serial monopoly. We assume that consumers have a quasilinear utility and that willingness to pay for today's consumption $y$ depends on the past consumption $x$, and is given by the function $u(x,y)$. When consumptions are substitutes,  $u$ is strictly submodular; when they are complements, $u$ is strictly supermodular.

We allow sellers to offer unrestricted non-linear prices and analyze the nature of inter-temporal competition. We consider a simple two-period model. In the benchmark case, where  consumption today is observed by the tomorrow's seller, our intuitions are confirmed---the first period seller  induces over-consumption relative to the efficient allocation when $u$ is submodular, and under-consumption when $u$ is supermodular. The intuition is straightforward: since the future seller will extract the buyer's surplus, the monopolist today seeks to induce the consumption level that maximizes the consumer's utility when she exercises her outside option tomorrow.\footnote{This finding is reminiscent of the literature on long-term contracts, e.g. \cite{DiamondMaskin79} or \cite{AghionBolton87}, where payoffs and utility are time separable. Here, contracts are short-term, but future utility depends on current consumption.}  This is  larger than the efficient amount when $u$ is submodular, and  smaller when $u$ is submodular.

Our main focus however, is on the more realistic case where neither the consumer's past consumption nor the past price offers are observed by the current seller. Our first result is that there cannot be a pure strategy equilibrium where consumption is deterministic, both when $u$ is submodular and when it is supermodular. 
The intuition for this is both simple and subtle: if first period consumption is deterministic and second period purchases are positive,  then the consumer must be indifferent between the second-period seller's offer and her outside option, since the former fully extracts the consumer's surplus. 
However, in this case, seller 1 and the consumer can increase their joint surplus, either by increasing first period consumption or by reducing it.  More generally, if the consumer is indifferent between her consumption level and seller 2's offer in \emph{any equilibrium}, pure or mixed, then seller 2 must exclude this consumer, i.e. seller 2's offer to this consumer must be  zero. 
Consequently, in any equilibrium, first period consumption choices must be random, thereby generating private information and informational rents for consumers in the second period, as well as the required exclusion for the consumers with the lowest marginal willingness to pay. 
Even if consumers are ex ante identical, the first period seller offers  a large set of quantities, giving rise to  ex-post taste heterogeneity.

Our main finding is that equilibrium outcomes are essentially unique.\footnote{All equilibria have the the same distribution over first and second period consumptions. They differ (possibly) only in terms of the distribution of payoffs between the first period seller and the consumer.} The first period monopolist offers a large menu, which ranges between the efficient quantity and that chosen in the observable consumption case. The consumer, who is indifferent between all bundles in the menu, chooses an item according to a continuous distribution with full support. This induces an endogenous screening problem in the second period, since the consumer has private information about her past consumption. We find that the consumer and the second-period seller benefit from unobservability, whereas the first-period seller loses, as compared to the observable consumption benchmark. Furthermore, consumer privacy unambigiously increases total welfare. 

The remainder of this paper is organized as follows. 
Section~\ref{sec:relatedliterature} discusses the related literature. 
Section~\ref{sec:model} sets out the model, and examines the case where the second period seller observes past consumption. 
Section~\ref{sec:privatepast} analyzes private transactions and establishes existence and essential uniqueness of an equilibrium where consumer heterogeneity arises endogenously, and examines how privacy affects welfare and consumer surplus. The final section concludes. All missing proofs can be found in the appendix.

\section{Related literature}\label{sec:relatedliterature}

Consumer privacy has become an important issue in the era of electronic records, big data and Internet shopping---see
\cite{AcquistiTaylorWagman16} for  a comprehensive survey of the topic. 
If past consumption decisions are observable, this may allow firms to identify the consumer's persistent type, thereby creating opportunities for price discrimination. 
When the consumer interacts repeatedly with the same firm,  \cite{Taylor04} shows that a naive consumer may be exploited by the firm. However, if the consumer is sophisticated, then the firm may want to commit to not utilizing personal data, in order to avoid the ratchet effect, an argument that is also made by \cite{Villas-Boas04}. In our setting, firms will not voluntarily make such a commitment, since making transactions public increases their monopoly power.  \cite{FudenbergTirole00} study behavior-based price discrimination, where the consumer's current choices reveals her relative preference for different brands. 

Our model differs considerably from this literature, since we assume ex ante homogeneity of consumers. Heterogeneity therefore arises only because of differences in past consumption choices, and these are endogenous. Thus conceptually, our work is more closely related to models with hidden actions, e.g.  work on static moral hazard with renegotiation (see \citet{FudenbergTirole90} and \citet{Ma91}). Similarly,  \citet{Gonzalez04} and \cite{Gul01} analyze the hold-up problem with unobservable investment. 

\citet{CalzolariPavan06} examine how upstream firms sell information about consumers' preferences and choices to downstream firms in an environment with indivisible goods and binary consumers' valuations. They characterize the conditions under which the upstream firms offer the consumers full privacy.

The inefficiencies we highlight  relate to the literature on long-term bilateral contracts in a multilateral environment. \citet{DiamondMaskin79} and \citet{AghionBolton87} show that 
a buyer-seller pair today may induce inefficiency  via long term contracts, in order to extract surplus from a future seller. 
Our contracts are static and the dynamics are induced
by the agent's preferences. Furthermore, our focus on private transactions differs from this literature, which assumes that the future seller observes the past contract. 

Our model bears some formal similarity with models of common agency
(\citet{BernheimWhinston86}; \citet{MartimortStole02}). The principals in these models correspond to our
sellers, and the agent to the consumer. A key difference is that 
the  consumer's  decisions are sequential in our model. When a consumer
receives an offer from a seller today, she does not have the option of
revising her purchases yesterday. We discuss this difference in more detail at the end of section \ref{sec-pse}. Whereas common agency models have a
plethora of equilibria and use refinements such as truthfulness to single
out a few, we find that equilibrium outcomes are essentially unique.

\section{The model}\label{sec:model}

The consumer, who lives for two periods, visits seller 1 in the first period
and seller 2 in the second period. Her utility is
\begin{equation*}
u(x,y)-p-q,
\end{equation*}where $x$ and $y$ is consumption in the first and second period
respectively, and $p$ and $q$ are the payments made to sellers $1$ and 2
respectively. The value of consumption in the second period depends on the
level of consumption in the first period. We assume that $u$
is strictly increasing, strictly concave and twice continuously
differentiable. We assume throughout that either \textbf{A1} or \textbf{A2} holds:
\begin{enumerate}
\item[\textbf{A1}:] $u(x,y)$ is strictly supermodular.
\item[\textbf{A2:}] $u(x,y)$ is strictly submodular.
\end{enumerate}
Given this assumption, we will dispense with the qualifier ``strictly" in the remainder of this paper.\footnote{That is, the statement ``$u$ is submodular'' should be read as ``$u$ is strictly submodular".}

Each seller has constant marginal cost $k$. Since a seller interacts with the consumer only for one period, he seeks to maximize his flow profit. We allow each seller to choose non-linear prices. Thus seller 1 chooses  an arbitrary lower semi-continuous 
function $p:\mathbb{R}_{+}\to \mathbb{R}_{+}$, where $p(x)$ is a price  for a bundle of size $x$, and similarly, seller 2 chooses a lower semi-continuous function $q:\mathbb{R}_{+}\to \mathbb{R}_{+}$. 

The socially efficient level of consumption $(x^{\ast
},y^{\ast})$ is defined as follows. Define $y^{\ast}(x)$ as the value of $y $ that solves $u_{1}(x,y)=k,$ if this equation has a positive
solution, and zero otherwise$.$ Since $u_{1}$ is strictly decreasing in $y,$
there is a unique value $y^{\ast}(x).$ 
This defines the function $y^{\ast}(x)$, and consider $x$ values such that $y^{\ast}(x)>0$; on this range, $y^{\ast}(.)$ is differentiable,  
strictly increasing when $u$ is supermodular and strictly decreasing when $u$ is submodular. 
Let $x^{\ast}$ be the value of $x$
that solves $u_{1}(x,y^{\ast}(x))=k$, and let $y^{\ast
}:=y^{\ast}(x^{\ast}).$ Since $u$ is strictly concave, there is unique
solution, so that the socially efficient level of consumption $(x^{\ast
},y^{\ast})$ is unique.

We will consider two distinct information structures, which differ only in the information observed by firm 2.
\begin{itemize}
\item First period consumption is observed by firm 2.
\item Transactions are private and firm 2 does not observe either offers or choices. 
\end{itemize}

Our notion of equilibrium is  sequential equilibrium: we require that the buyer's choices are sequentially rational. Further,  when seller 1's offer is unobserved by seller 2, then any deviation by seller 1 does not affect the buyer's beliefs about the pricing scheme that will be offered by seller 2.

\subsection{Observable consumption}

Consider  the situation where first-period consumption is
perfectly observed by the second-period seller. The pricing scheme offered by seller 1 may or may not be observed by seller 2, this does not matter;
 the price paid by the consumer in period one does not affect incentives in period two since the consumer has quasi-linear utility. 
In this case, seller 2 acts a Stackelberg leader. 
The quantity that seller 1 sells is chosen so as to maximize the joint payoff of seller 1 and the consumer, 
over the two periods, since seller 1 can extract all the surplus.  Since seller 2 is a monopolist, in period 2, the consumer gets no surplus in the second period, and her payoff is equal to that from choosing the outside option, i.e. it equals $u(x,0)$.\footnote{Seller
2 chooses $y=y^{\ast}(x),$ and sets $q=u(x,y^{\ast}(x))-u(x,0).$ Thus the
consumer's second period payoff $u(x,y^{\ast}(x))-q$ equals $u(x,0),$ since   payments made to seller 1 are already sunk.} The consumer's value  from consuming $x$ equals
\begin{equation*}
u(x,0)-p.
\end{equation*}

If the consumer does not buy from the first period monopolist, her overall
payoff equals $u(0,0)$ (since her second period continuation payoff after any $x$
always equals $u(x,0).$ Since the first period seller optimally sets $p$ so
that
\begin{equation*}
p(x)=u(x,0)-u(0,0),
\end{equation*}
the bundle that maximizes the first period profits when consumption is
observable, $x^{o},$ satisfies

\begin{equation*}
u_{1}(x^{o},0)=k.
\end{equation*}

We shall assume throughout:
\begin{enumerate}
\item[\textbf{A3}:]$y^{o}:=y^{\ast}(x^{o})>0.$
\end{enumerate}
This assumption ensures that the first seller cannot serve the customer alone and the second seller plays a meaningful role in this market. An alternative formulation for this assumption is $u_2(x^{o},0)>k$.

\begin{proposition}\label{prop-obs}
Suppose that first period consumption $x$ is observed by the second period seller. 
If  $u$ is  submodular, then  $x^{o}>x^{\ast}$ so that
the first period monopolist induces excessive consumption relative to the
first best. If $u$ is  supermodular, then $x^{o}<x^{\ast},$ so that 
the first period seller induces underconsumption relative to the first best.
Second period consumption is always (conditionally) efficient. 
\end{proposition}

The proof of this  proposition is intuitive. First-period inefficiency arises from the fact that  the first period seller maximizes 
the consumer's second period outside option, $u(x,0)$, rather
than the consumer's utility associated with her actual consumption, $u(x,y^{o})$. When $u$ is supermodular, the maximizer of $u(x,0)$ is strictly less than the maximizer of $u(x,y^{\ast}(x))$, since $y^{\ast}(x^{o})>0$.  When $u$ is submodular, the maximizer of $u(x,0)$ is strictly less than $x^{\ast}$. 

Finally, note that firm 1 need not restrict the consumer to a singleton menu. 
It could, for example, offer a two-part tariff, where the price per unit is set at marginal cost $k$, and the fixed fee $A$ is such that $A+kx^o$ equals $p$, the optimal price for $x^o$ under the singleton menu.
The consumer will choose $x^o$, and so the outcome will be the same as in the proposition above. 

We note that both $x^{\ast}$, the efficient level,  and $x^o$, the consumption level when it is observable,  play an important role in the analysis when transactions are entirely private.

\section{Private Transactions}\label{sec:privatepast}

We now analyze  the main specification of our model, where  seller 2 can observe neither the offer made by seller 1 nor the consumer's choice in the first period.

\subsection{Non-existence of a pure strategy equilibrium}\label{sec-pse}

Our first result
is that there does not exist a pure strategy equilibrium in this case,
either in the supermodular or the submodular case.
To gain some intuition for this result, consider the case where $u$ is submodular.
Let us see why consuming $x^o$ in the first period is not an equilibrium. 
Suppose firm 2 believes that the consumer has indeed consumed $x^o$. 
In this case, firm 2 would offer $y^o$, at a price $q=u(x^o,y^o)-u(x^o,0)$. 
Given that this is firm 2's offer, the payoff of the coalition consisting of firm 1 and the consumer (Figure~\ref{fig:x0}, dashed line) 
from any $x<x^o$ is given by 
\[
u(x,y^{\ast}(x^o)) - q - kx.
\]
Since submodularity of preferences implies that the consumer strictly prefers $(y^{\ast}(x^o),q)$ to her outside option when $x<x^o$.
The derivative of the above expression with respect to $x$ is \emph{negative} at $x=x^o$, since $x^o$ maximizes $u(x,0)-kx$ (Figure~\ref{fig:x0}, solid line).
In other words, the Stackelberg outcome fails to be an equilibrium outcome when the first period action is not observable for a familiar reason---it is not a best response to the firm 2's action. 

This raises the question, why is there not a pure strategy Nash equilibrium where the consumer chooses some $\tilde{x}<x^o$?
In such a candidate equilibrium, firm 2 will offer $y^{\ast}(\tilde{x})$ at a price $q=u(\tilde{x}, y^{\ast}(\tilde{x}))-u(\tilde{x},0)$. Suppose now that firm 1 offers some $x\in (\tilde{x},x^o]$. If the consumer accepts this offer, then it will optimal for her to take the outside option in the second period. Consequently, the coalition consisting of the firm 1 and the consumer will get a payoff of 
\[
u(x,0) -kx.
\]
The derivative of this payoff with respect to $x$ is strictly positive, since $x<x^o$ (Figure~\ref{fig:xtilde}, dashed line). In other words, there cannot be an equilibrium pure consumption level below the Stackelberg level, since then there is an incentive to deviate upwards. 

\begin{figure}
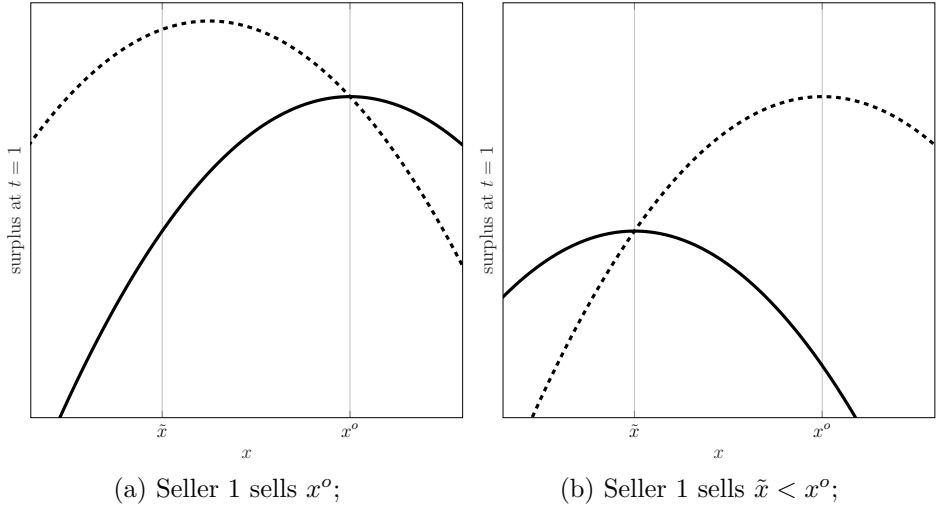

    \centering
    \subfloat[Seller 1 sells $x^{o}$;\label{fig:x0}]{
        \input{fig_x0.tex}}
	\subfloat[Seller 1 sells $\tilde{x}<x^{o}$;\label{fig:xtilde}]{
        \input{fig_xtilde.tex}}
\caption{Profitable deviations.}\label{fig:x}
\end{figure}

The fundamental problem is as follows. On the one hand, first period consumption $x$ must maximize $u(x,y)-kx$, since the consumer's actual consumption in equilibrium is the pair $(x,y)$. On the other hand, $x$ must maximize $u(x,0)-kx$, since the consumer's second  payoff in a pure strategy equilibrium must equal $u(x,0)$, the payoff that she gets from the outside option. 
Since no single value of $x$ can solve both these maximization problems when one has either strict supermodularity of submodularity, there cannot be a pure strategy equilibrium. 

\begin{proposition}\label{prop:nopurestr}
If $u(x,y)$ is strictly submodular or strictly supermodular, and assumption A3 is satisfied, there does not
exist an equilibrium where the consumption in period 1 is deterministic.
\end{proposition}
\begin{proof}
See appendix \ref{sec:nopurestr}.
\end{proof}

This argument can be generalized beyond pure strategies: in equilibrium, the consumer cannot be indifferent between two different menu items offered by seller 2, as we show in Lemma~\ref{lemma:deviations} in Appendix \ref{sec:nopurestr}.

Since our model bears some  resemblance to the common agency model  under complete information \citep{BernheimWhinston86}, it may be illuminating to see why the truthful equililibrium of that model fails to be an equilibrium in our setting. In the common agency version of our model, both sellers choose non-linear price schedules simultaneously; the consumer observes both schedules before choosing  from each. In a truthful equilibrium, both sellers offer two-part tariffs with the per-unit price equal to marginal cost and the consumer chooses the efficient bundle $(x^*,y^*)$.  When $u$ is submodular,\footnote{The argument can be made similarly for the case when  $u$ is supermodular.}  the fixed tariff that firm 2  charges satisfies:\footnote{Symmetrically, $A_1 = u(x^*,y^*)-u(0,y^\dagger)-k[y^*-y^\dagger] - kx^*$, where $y^\dagger$ is the maximizer of $u(0,y)-ky$.}
\[
A_2  = u(x^*,y^*)-u(x^o,0)-k[x^*-x^o] - ky^*.
\]
Firm 2 cannot charge more than this since the consumer will respond by rejecting firm 2's offer and increasing her consumption from firm 1 to $x^o>x^*$. 

In our model, the firms effectively choose actions simultaneously, since firm 2 does not observe firm 1's actions. 
However,  the consumer must choose from firm 1 before she sees the offer from firm 2. Consequently, in a candidate equilibrium  where firm 1 chooses $(A_1,k)$, and the consumer chooses $x^*$, firm 2 can deviate to 

\[
A^{'}_2 = u(x^*,y^*)-u(x^*,0)- ky^*. 
\]
The difference $A^{'}_2-A_2$ equals 
\[
[u(x^o,0)-kx^o] -  [u(x^*,0)-kx^*]>0,
\]
since $x^o$ is the unique maximizer of $u(x,0)-kx$.   
In other words, the common agency profile fails to be an equilibrium since the consumer cannot revise her choices at firm 1 after observing firm 2's offer.

\subsection{Endogenous screening: An overview}\label{sec:informal}

Having established that first period consumption must be random, we now construct an equilibrium with the following features. 
Seller 1  charges a  two-part tariff $(A,k)$,  where $A$ is the fixed fee, and the price per unit equals the marginal cost $k$. 
The consumer chooses period 1 consumption randomly, according to a continuous c.d.f. $F$ that has support the interval  $[\underline{x},\overline{x}]$.  
Since first period consumption $x$ affects the  consumers willingness to pay for the second period bundles, seller 2 offers an optimal  screening menu. It is convenient to think of this menu in the form of a direct mechanism $(\hat{y}(x),q(x))$, where the second period consumption $\hat{y}(x)$ is strictly decreasing in $x$. 
Under this menu, a consumer with first period consumption $x$ gets second period indirect utility $U(x)$,\footnote{$U(x):=u(x,\hat{y}(x))-q(x)$, and  the consumer's overall payoff equals $U(x)-kx-A$.} where $U(x)-kx$ is constant on the interval $[\underline{x},\overline{x}]$.  Consequently, due to the two-part tariff in the first period, with marginal price $k$, the consumer is indifferent between all bundles in this interval, and it is optimal for her to randomize according to the distribution $F$. 
Finally, equilibrium is ``essentially unique". 
In the supermodular case, the equilibrium outcome is indeed unique.  In the submodular case, if the second period seller caters to first period consumption levels smaller   than  $\underline{x}$, this may affect the division of payoffs between firm 1 and the consumer, by affecting the first-period outside option  of the consumer.  

Let us now examine how this randomization helps overcome the fundamental impossibility that arose with a pure strategy equilibrium. Consider, for illustrative purposes, the submodular case. Higher values of $x$ correspond to ``lower'' types, and the worst type is the consumer who has consumed $\overline{x}$.  Such a consumer is held to her outside option, and thus must be indifferent between accepting the offered second period bundle designed for her, $(\hat{y}(\overline{x}),q(\overline{x}))$ and consuming $y=0$. If $\hat{y}(\overline{x})$ was strictly positive, then the fundamental contradiction that arose in the pure strategy case would also arise here, since $\overline{x}$ cannot simultaneously maximize $u(x,\hat{y}(\overline{x}))-kx$ and $u(\overline{x},0)-kx$. However, if 
$\hat{y}(\overline{x})=0$, then no contradiction arises, and indeed, that is the resolution to this problem.
In other words, the induced distribution $F$ ensures the exclusion of the consumer who has consumed $\overline{x}$.  

More generally, given the allocation $\hat{y}(\tilde{x})$ for any type $\tilde{x}$, $\tilde{x}$ must maximize $u(x,\hat{y}(\tilde{x}))-kx$.
That is, given the second period consumption, the first period bundle chosen by the consumer must maximize her payoff, net of the marginal cost.  In other words, the second period consumption $\hat{y}(x)$  must satisfy:
\begin{equation}
u_1(x,\hat{y}(x)) =k, \forall x\in [\underline{x},\overline{x}].\label{eq:allocation-rule}
\end{equation}
Observe that this uniquely pins down the second period consumption $\hat{y}(x)$, and ensures that it is strictly decreasing when $u$ is submodular, and strictly increasing when $u$ is supermodular.

To summarize, the critical features for any equilibrium are as follows:

\begin{enumerate}
\item\label{item:opt1} $U(x)-kx$ is constant for every $x\in [\underline{x},\overline{x}]$. This ensures that the consumer and seller 1 are indifferent as to which element of $[\underline{x},\overline{x}]$ the consumer chooses.

\item The induced distribution of the first period consumption, $F,$ is such that seller $2$ finds it optimal to offer $U(x)$ for each $x\in [\underline{x},\overline{x}]$.

\item The second period consumption $\hat{y}(x)$ satisfies \ref{eq:allocation-rule}, and  is therefore strictly decreasing in the submodular case, and
strictly increasing in the supermodular case.

\item Finally, the endpoints of the interval $[\underline{x},\overline{x}]$ are pinned down by the
characteristics of the solution to the monopoly screening problem. Since
there is no distortion at the top, the second period consumption of the
highest type---e.g., \underline{$x$} in the submodular case---must be
optimal given $\underline{x}.$ Combined with point (\ref{item:opt1}) above, this implies
that $\underline{x}=x^{\ast},$ the first best level of consumption (when
utility is supermodular, the highest type corresponds to $\overline{x}$ which
must equal $x^{\ast})$. Since there is no informational rent at the
bottom---e.g., for type $\overline{x}$ in the submodular case---her consumption level must maximize the joint payoff of the consumer and seller 1 given
that she takes the outside option $0$ in the second period. This implies $\overline{x}=x^{o}.$
Thus, the first period consumptions span the range between first best and the equilibrium consumption in the case when the past history is
observable, while second period consumptions lie between $0$
and the first best consumption, $y^{\ast}.$
\end{enumerate}

\subsection{Equilibrium characterization}\label{sec:equilibrium}

In this section we formalize the ideas presented in Section~\ref{sec:informal}. We begin with the characterization of the continuation payoff in the second period and then we provide conditions that pin down the distribution of consumption in the first period.

Proposition~\ref{prop:nopurestr} has established that in any equilibrium, the first period consumption must be random. Let $X$ denote the support of the equilibrium
distribution of the first period consumption---$X$ is a closed set, by
definition. We shall also assume that every bundle in $X$ is offered and
chosen by the consumer.\footnote{That is, we assume that the set of chosen bundles is closed, so that every $x\in X$ has an associated pair $(\hat{y}(x),q(x))$ in the menu. This assumption is inessential, but simplifies the statement of some results.} Note that $X$ cannot contain $0$---in this case,
seller 1's profits must equal zero, and this cannot be optimal for seller 1 and
the consumer.

Denote a menu offered by seller 2 in equilibrium by $(\hat{y}(x), q(x))_{x\in X}$. The second period indirect utility of  the consumer after choosing bundle $x$ in the first period is
\[
U(x):=u(x,\hat{y}(x))-q(x).
\]
The sum of the payoffs for the consumer and seller 1 if the former chooses $x$ is
\[
\Sigma(x):=U(x)-kx.
\]

First, we extend $U$ so that it is defined on an open interval $I\supseteq X$
rather than just the chosen points, $X,$ where $I\subset (0,\infty ).$ For $z\in I-X,$ let
\[
U(z):=\sup_{x\in X}\{u(z,\hat{y}(x))-q(x)\}.
\]
Thus, $U$ is specified by prescribing optimal choices for all non-chosen types, and every
point in $X$ lies in the interior of $I.$

\begin{lemma}\label{lemma:differentiable}
$U(x)$ is differentiable at every chosen $x\in X.$
\end{lemma}
\begin{proof}
See appendix \ref{app:diff}.
\end{proof}

\begin{remark}
The property that $U$ is differentiable on $X$, the set of types that are chosen in equilibrium, follows from the endogeneity of types. With exogenous
types, it is well known that $U$ need not be everywhere differentiable.
Indeed, this observation is more general than the specific context of our
model.
\end{remark}

It is standard in theory of incentives that single-crossing and incentive
compatibility implies weak monotonicity. However, lemma \ref{lemma:differentiable}
allows a stronger result.

\begin{lemma}\label{lemma:monotonicity}
$\hat{y}(x)$ must satisfy
\begin{equation*}
u_{1}(x,\hat{y}(x))=k.
\end{equation*}Moreover, $\hat{y}(x)$ is strictly decreasing (resp. increasing) in $x$ if $u$
is submodular (resp. supermodular).
\end{lemma}

\begin{proof}
Since $U$ is differentiable at $x\in X,$ if $x$ maximizes $\Sigma(.),$ it must
satisfy
\begin{equation*}
\Sigma ^{\prime }(x)=u_{1}(x,\hat{y}(x))-k=0.
\end{equation*}Consider a case of submodular utility---i.e., $u_{21}<0.$ If $x>\tilde{x}$ then $\hat{y}(x)$ must be
strictly less than $\hat{y}(\tilde{x}),$ or otherwise the expression for $\Sigma^{\prime}(.)$ above will be strictly negative. Similarly, in the case of supermodular utility---i.e., if $u_{21}>0$,
 $\hat{y}(x)$ must be strictly greater than $\hat{y}(\tilde{x})$.
\end{proof}

Let $\underline{x}$ denote the minimal element in $X$ and $\overline{x}$ the maximal
element. The following lemma shows that if individual rationality is
satisfied for type $\overline{x}$ in the submodular case, then it is satisfied
for every other type---although familiar, the result is not immediate since
the outside option $u(x,0)$ is type dependent.  A similar result is true  for the case of supermodular utility.

\begin{lemma}
If $u$ is submodular, $U(x)-u(x,0)\geq U(\overline{x})-u(\overline{x},0)$ for all $ x\in X, x\neq\overline{x}$. Moreover,
under any profit maximizing second period contract, $U(\overline{x})=u(\overline{x},0),$
and the individual rationality constraint binds for type $\overline{x}.$ If $u$ is supermodular, the individual rationality constraint binds for type 
\underline{$x$}, and is slack for every other type in $X$.
\end{lemma}

\begin{proof}
For $x<\overline{x},$ since type $x$ can pretend
to be $\overline{x},$  incentive compatibility implies that
\begin{equation*}
U(x)\geq U(\overline{x})+u(x,\hat{y}(\overline{x}))-u(\overline{x},\hat{y}(\overline{x})).
\end{equation*}

Since $U(\overline{x})\geq u(\overline{x},0),$

\begin{equation}
U(x)-u(x,0)\geq \lbrack u(\overline{x},0)-u(x,0)]-[u(\overline{x},\hat{y}(\overline{x}))-u(x,\hat{y}(\overline{x}))],  \label{IR-ineq}
\end{equation}

which is non-negative since $u$ is submodular and $\hat{y}(\overline{x})\geq 0.$

If $U(\overline{x})>u(\overline{x},0),$ then a menu $(\hat{y}(x),q(x))$ cannot be profit maximizing, since a uniform reduction in payoffs $U(x)$
by $U(\overline{x})-u(\overline{x},0),$ achieved by raising $q(x)$ by the same amount,
preserves incentive compatibility and increases profits. To obtain the same result for the case of supermodular $u$, replace $\overline{x}$ by \underline{$x$} in the above
argument.
\end{proof}

The following two lemmata identify $\underline{x}$ and $\overline{x}$. Recall that one of the bounds is identified using the fact that the highest type's consumption in the second period is efficient. In order to identify the other bound, we consider possible deviations by seller 1 and establish that the lowest type has to consume zero in the second period.

\begin{lemma}\label{lemma:xo}
If $u$ is submodular, $\overline{x}$ equals the value of $x$ that
maximizes $u(x,0)-kx$---i.e., $\overline{x}=x^{o},$ and $\hat{y}(\overline{x})=0.$ If $u$ is supermodular, $\underline{x}$ equals the value of $x$ that
maximizes $u(x,0)-kx$---i.e., $\underline{x}=x^{o},$ and $\hat{y}(\underline{x})=0.$
\end{lemma}

\begin{proof}
If $u$ is submodular, since the second period participation constraint binds
for the highest value of $x$ that is offered by seller 1 and accepted by the
consumer, the consumer is indifferent between $\hat{y}(\overline{x})$ and zero in the second period. By Lemma~\ref{lemma:deviations}, there exists a profitable deviation for seller 1 unless $\hat{y}(\overline{x})=0$. Consequently, lemma~\ref{lemma:monotonicity} implies that $\overline{x}$ must equal the value of $x$ that maximizes
\begin{equation*}
u(x,0)-kx,
\end{equation*}
so that $\overline{x}=x^{o}$. The proof for the case of supermodular utility is identical.
\end{proof}

\begin{lemma}\label{lemma:x*}
If  $u$ is submodular, then $\underline{x}=x^{\ast}$; and if it is supermodular, then $\overline{x}=x^{\ast}$. In either case, \ $\hat{y}(x^{\ast})=y^{\ast}.$
\end{lemma}

\begin{proof}
Recall that $y^{\ast}(x)$ denotes the first best second period quantity
conditional on any level of the first period consumption $x$. Suppose that $u$ is submodular. On one hand, since there is no distortion at the top in the second period screening problem, seller 2 must offer $y^{\ast}(\underline{x})$ to the consumer who consumed $\underline{x}$ in the first period. On the other hand, Lemma~\ref{lemma:monotonicity} establishes that $\underline{x}$ must satisfy
\begin{equation*}
u_{1}(\underline{x},\hat{y}(\underline{x}))=k.
\end{equation*}These two conditions imply\begin{equation*}
u_{1}(\underline{x},y^{\ast}(\underline{x}))=k,
\end{equation*}which means that $(\underline{x},y^{\ast}(\underline{x}))$ satisfies the
conditions for the first best allocation. The first best allocation is
unique, therefore $\underline{x}=x^{\ast}.$ When $u$ is supermodular,
the "top" corresponds to $\overline{x},$ and the rest of the argument is
the same.
\end{proof}

To summarize, the characterization in the above lemmata imply that $\overline{x}=x^{o}$ and $\underline{x}=x^{\ast}$ in the case of submodular utility. When $u$ is
supermodular, $\overline{x}=x^{\ast}$ and $\underline{x}=x^{o}$. 

To complete the description of the equilibrium, it remains to specify  the distribution of first period consumption $F$ that induces the second period consumption $\hat{y}(x)$ and consumer indirect utility $U(.)$ on the interval $[\underline{x},\hat{x}]$, and this is set out in the following theorem:

\begin{theorem}
\label{t:eq} There exists an equilibrium in which

\begin{enumerate}
\item Seller 1 offers a two-part tariff. The entry fee equals to the seller
1's value added in the socially efficient consumption stream:
\begin{equation*}
u(x^{\ast},y^{\ast})-kx^{\ast}-u(0,y^{\ast}).
\end{equation*}
The per-unit price equals to the marginal cost $k$.

\item Seller 2 offers a menu that includes every bundle in $[0,y^{\ast}]$.
The bundles in this menu are indexed by the first period consumption $x$.
The price of a bundle $\hat{y}(x)$ is
\begin{equation*}
q(x)=u(x,\hat{y}(x))-kx-[u(\overline{x},0)-k\overline{x}]
\end{equation*}

\item In the first period, the consumer randomly chooses the bundle
according to a distribution $F$. In the second period, she chooses a
consumption $\hat{y}(x)$ where $x$ is her first period consumption.

If $u$ is submodular, the support of the distribution $F$ is $[x^{\ast},x^o]$
and
\begin{equation*}
F(x)=\exp \left[ \int\limits_{x}^{\overline{x}}\frac{u_{21}(z,\hat{y}(z))}{u_{2}(z,\hat{y}(z))-k}dz\right] .
\end{equation*}

If $u$ is supermodular, the support of the distribution $F$ is $[x^o,x^{\ast}]$ and
\begin{equation*}
F(x)=1-\exp \left[ \int\limits_{\underline{x}}^{x}\frac{u_{21}(z,\hat{y}(z))}{k-u_{2}(z,\hat{y}(z))}dz\right] .
\end{equation*}
\end{enumerate}
\end{theorem}

Verification that the above distribution indeed accomplishes the task is set out in Appendix \ref{app:main}, which completes the proof of the theorem.

\subsection{Uniqueness of equilibrium outcomes}

All equilibria in this model have several common features. If the utility is supermodular, any equilibrium must have the same \emph{outcome} as the equilibrium described in Theorem~\ref{t:eq}---that is the distribution of consumptions and the payoffs of the sellers and the consumer must be the same.\footnote{It is possible to generate this equilibrium outcome in several different ways. For example, seller 1 can perform the randomization and present the  consumer with a single  $x$, where $x$ has distribution $F$.} If the utility is submodular, there is a continuum of equilibria outcomes;  however these outcomes only differ  because they distribute 
differently the surplus between seller 1 and the consumer---seller 2's payoff and the distribution of consumptions are invariant.
Since the allocation is invariant across these equilibria, we say that the equilibrium outcome is \emph{essentially unique}.

The multiplicity of payoff division in the submodular case arises because seller 2 can add items that are larger than $y^{\ast}$ to his menu and price any quantity above $y^{\ast}$ at  marginal cost.  
Every such item $(q,y)$ satisfies $y = y^{\ast}+\delta$ for some $\delta>0$ and $q = q(y^{\ast})+k\delta$.
Let $\tilde{y}$ denote the largest such item such that $u_{2}(0,\tilde{y})\geq k$. 
These items on the menu will not be chosen on equilibrium path. However, if the consumer decides not to consume in the first period, then $\tilde{y}$   becomes the most attractive on the  menu, and gives the consumer 
a strictly higher payoff than  $(q(y^{\ast}),y^{\ast})$ because $u_{2}(0,y^{\ast})>k$.
 Therefore,  the value of the first-period outside option is  increasing in $\tilde{y}$, but these items do not affect firm 2's payoff.
 By offering large $\tilde{y}$,  seller 2 reallocates the first-period surplus from seller 1 to the consumer, but does not modify the distribution of consumptions or his own payoff.

To formalize these ideas we characterize the objects that are invariant across all equilibria. We focus on the submodular case, since equilibrium outcomes are unique  when $u$ is supermodular.\footnote{Formally,  the results below apply  to both cases.} Fix an equilibrium $\sigma $ of the game. Let $V_{\sigma}$ denote the consumer's ex ante utility and let $\pi _{1,{\sigma}}$ denote
the expected profit of seller 1  in this equilibrium. Let $\Sigma_{\sigma}:=V_{\sigma}+\pi_{1,{\sigma}}$ denote the sum of payoffs of the consumer
and seller 1. Also, let $X_{\sigma}$ denote the set of the first-period consumptions chosen in equilibrium $\sigma$, and $U_{\sigma}(x)$ denote the information rent of the consumer after choosing $x$ in the first period of equilibrium $\sigma$. Let $\tilde{\sigma}$ denote the specific equilibrium constructed
in the previous section, the support of which is the largest possible set, $X_{\tilde{\sigma}} = [x^{\ast},x^{o}].$ Note, that by Lemmas~\ref{lemma:xo} and \ref{lemma:x*}, $X_\sigma\subset X_{\tilde{\sigma}}$.

\begin{lemma}\label{lemma:samepayoff}
For any equilibrium $\sigma:\Sigma_{\sigma} = \Sigma_{\tilde{\sigma}}$ and for  any $x\in X_{\sigma}: U_{\sigma}(x)=U_{\tilde{\sigma}}(x).$
\end{lemma}

\begin{proof}
By Lemma~\ref{lemma:xo} the maximal quantity offered and chosen in period one
equals $x^{o}$ in any equilibrium, and the informational rent that accrues
to the consumer is zero in this case. Since $x^{o}$ is in the support of
every equilibrium, for any equilibrium $\sigma$
\begin{equation*}
\Sigma_{\sigma}=u(x^{o},0)-kx^{o}=\Sigma_{\tilde{\sigma}}.
\end{equation*}For all $x^{\ast}\in X_{\sigma}$ the following holds in equilibrium
\begin{equation*}
\Sigma_{\tilde{\sigma}}=U_{\sigma}(x)-kx,
\end{equation*}therefore $U_{\sigma}(x)= U_{\tilde{\sigma}}(x)$.
\end{proof}

In the light of this proposition, we write $\Sigma $ and $U(x)$ for
the payoffs that arise in \emph{any} equilibrium. Let $F$ denote the
c.d.f. associated with $\tilde{\sigma},$ as defined in the previous section,
and let $f$ denote the associated density.

\begin{theorem}\label{t:uniqueness}
\begin{enumerate}
	\item If $u$ is supermodular, then the equilibrium outcome is unique.
	\item If $u$ is submodular, the equilibrium outcome is essentially unique. In every equilibrium,
	\begin{enumerate}
		\item the distribution of the first-period consumption is $F$ with the support on $[x^{\ast},x^{o}]$;
		\item the items of the second-period menu that are chosen on equilibrium path are $\{(q(x),\hat{y}(x))\}_{x\in [x^{\ast},x^{o}]}$; and
		\item the sum of the equilibrium payoffs of seller 1 and the consumer is $\Sigma$.
	\end{enumerate}
\end{enumerate}
\end{theorem}

\begin{proof}
The proof is in appendix \ref{app:uniqueness}.
\end{proof}

The only difference between the cases of sub- and supermodular utilities is the effect of second-period items that are not chosen by the consumer on equilibrium path. In the case of supermodular utility, these items make no difference. However, in the case of submodular utility, if the seller in the second period offers items $y>y^{\ast}$ in the menu (but makes them unattractive to the consumer who consumed $x\in[x^{\ast},x^{o}]$), he may make the outside option in the first period more valuable to the consumer. Such additions to the menu may affect the prices in the first period without affecting the equilibrium consumption.

\subsection{Privacy: Effects on  welfare, consumer surplus and profits}

\label{sec:payoffs}

Our analysis shows that  consumers' privacy dramatically affects equilibrium outcomes. 
The equilibrium outcomes under privacy are very different from the equilibrium outcome in the benchmark, in which past consumption is  observed by seller 2. 
This has efficiency and payoff consequences. We find that under privacy: 
\begin{enumerate}[(i)]
\item Social welfare is greater;
\item  Consumer's utility is greater and seller 2's profit is larger;

\item Seller 1's profit is lower.
\item the sum of seller 1's profit and consumer utility is invariant. 
\end{enumerate}

The most interesting of these results is that total welfare increases under consumer privacy. 
To understand this, fix attention on the submodular case. 
When transactions are public, there is excessive consumption in the first period, at $x^o$, in order to improve the second period outside option of the consumer, while second period consumption is conditionally efficient. 
Under private transactions, first period consumption becomes more efficient---indeed, it lies in the interval $[x^{\ast},x^o]$. 
Although second period consumption is no longer constrained efficient, the improvement in first period efficiency more than offsets this. 

Seller 2 is also better off under consumer privacy, since his ignorance protects him from being exploited as a Stackelberg follower. 
This more than offsets his consequent inability to fully extract consumer surplus.

The consumer is better off since privacy increases her outside option in the first period. 
Since her choices are not  observed by seller 2, the consumer can choose the outside option in period 1 without causing seller 2 to increase the prices in period 2. This increases the value of the first-period outside option as  compared to the benchmark with observable consumption.

Finally, the result that  the sum of seller 1's profit and consumer utility is the same as under observable past consumption may appear surprising. 
However, the intuition is straightforward  ---  this sum under privacy is the same regardless of the consumer's  first period choices, and so 
it equals the sum of payoffs when the consumer chooses $x^o$ in period one and her outside option in period two.\footnote{When  the past is observable, the consumer's overall payoff is the same whether she buys or does not buy in the second period.}

Let $\pi_i, i\in\{1,2\}$ denote the profits of the two firms under privacy, and let $V$ denote the consumer's payoff,  let $W$ denote social welfare.  
Let $\Sigma:=\pi_1+V$ denote the sum of payoffs of the consumer and the firm.  
We append superscript $o$ to each of these variables to denote their equilibrium values under the observable transactions benchmark---e.g. the consumer's payoff is $V^o$. 
The subscript $\sigma$---associates a variable with equilibrium  $\sigma$ under privacy.
We denote by $\tilde{\sigma}$ the specific equilibrium 
 characterized in Theorem~\ref{t:eq}, where the  second-period menu was minimal. 
\begin{theorem}\label{t:payoffs}
If $u$ is supermodular, then
\begin{enumerate}[(i)]
\item $\Sigma = \Sigma^{o}$;
\item $\pi_1<\pi_1^{o}$ and $V>V^{o}$; and
\item $W-W^o =\pi_2-\pi_2^o>0$.
\end{enumerate}

If $u$ is submodular, then for any equilibrium $\sigma$
\begin{enumerate}[(i)]
\item $\Sigma_\sigma = \Sigma^{o} $ $\forall \sigma$;
\item $\pi_{1,\sigma}\leq\pi_{1,\tilde{\sigma}}<\pi_1^{o}$ and $V_{\sigma}\geq V_{\tilde{\sigma}}>V^{o}$; and
\item $W_{\tilde{\sigma}}-W^o=W_{\sigma}-W^o =\pi_{2,\sigma}-\pi_2^o>0$.
\end{enumerate}

\end{theorem}
\begin{proof}
Note that the first-period surplus, $\Sigma$, can be evaluated at any point in the support. In particular, at $x^{o},$ the consumer takes her outside option in the second period, and so

\begin{equation*}
\Sigma =u(x^{o},0)-kx^{o}.
\end{equation*}

This is \emph{identical} with the total payoff of the seller 1 and the consumer when
consumption is observable---although the consumer purchases $y^{o}>0,$ seller
2 appropriates the the difference $u(x^{o},y^{o})-u(x^{o},0)$, and hence the consumer's continuation payoff is $u(x^{o},0).$ We turn to the distribution of the total payoff between the two
parties in the two cases. In the observable case, the consumer's payoff
equals
\begin{equation*}
V^{o}:=u(0,0).
\end{equation*}

In the unobservable case, the results now differ depending on whether $u$ is
supermodular or submodular. So we consider these in turn.

When $u$ is supermodular, the consumer who chooses the outside option in the first period, chooses the item $\hat{y}(x^{o})=0$ in the second period, and therefore gets a total payoff
\begin{equation*}
V = u(0,0).
\end{equation*}
This is exactly equal to $V^{o},$ and hence unobservability has no
distributional effect on the first period payoffs when $u$ is supermodular.

When $u$ is submodular, there is a continuum of equilibria that differ by the value of the outside option in the first period. Consider the equilibrium with the smallest such value---i.e., the equilibrium $\tilde{\sigma}$  characterized in Theorem~\ref{t:eq}. In this equilibrium, if the consumer chooses the outside option in the first period, she
buys $\hat{y}(x^{\ast}) = y^{\ast}$ in the second period and, therefore, gets a total payoff
\begin{equation*}
\underline{V} := V_{\tilde{\sigma}} = \left[ u(0,y^{\ast})-u(x^{\ast},y^{\ast})\right] +u(x^{\ast},0).
\end{equation*}

The difference in payoffs is
\begin{equation*}
\underline{V}-V^{o}=[u(x^{\ast},0)-u(0,0)]-\left[ u(x^{\ast},y^{\ast})-u(0,y^{\ast})\right] =\pi _{1}^{o}-\overline{\pi}_{1}>0,
\end{equation*}
where $\pi_{1}^{o}$ denotes seller 1's profits in the observable case. The
second equality in the above follows since the total payoff $\Sigma$ is
equal in the two cases. The strict inequality arises since $u$ is strictly
submodular.

Now consider the equilibrium with the largest value of the first-period outside option. Using the same argument we obtain that
\begin{equation*}
\overline{V} := \left[u(0,y^{\ast}(0))-u(x^{\ast},y^{\ast})\right]-k(y^{\ast}(0)-y^{\ast}) +u(x^{\ast},0)>\underline{V}.
\end{equation*}
and $\overline{\pi}_1>\underline{\pi}_1$. For every $V\in[\underline{V}, \overline{V}]$, there exists an equilibrium in which the consumer's payoff is $V$ and seller 1's profit is $\Sigma-V$.

We conclude that, in any equilibrium, the consumer is strictly better off when
consumption is unobservable, and seller 1 is strictly worse off to exactly the
same extent.

Since $\Sigma$ is the same across all equilibria including the benchmark case of the observable past, the gain (loss) of seller 2 from unobservability of past consumption equals to the increase (decrease) of the social welfare. Moreover, this gain (loss) is the same for all equilibria because the equilibrium distribution of consumption is the same.

To see that $\pi_2>\pi_2^o$, note that if seller 2 offers a single item $(q^o,y^o)$, it will be accepted with probability 1 because every consumer's type is better (in terms of marginal willingness to pay) than $x^o$. The fact that such a menu is not offered implies that $\pi_2>\pi_2^o$.
\end{proof}

The equilibrium consumption is distorted by the intertemporal competition between the sellers in an unusual way. If the utility is submodular, the consumer always over-consumes in the first period and under-consumes in the second. The realized social welfare is monotone in the first period consumption.

\begin{remark}
Equilibrium social welfare conditional on first period consumption $x$ is
decreasing in $x$ when $u$ is submodular, and increasing when $u$ is
supermodular.
\end{remark}

\begin{proof}
Equilibrium social welfare conditional on first period consumption $x$ is
\begin{equation*}
W(x)=u(x,\hat{y}(x))-k(x+\hat{y}(x))
\end{equation*}Taking a derivative, we obtain
\begin{eqnarray*}
W^{\prime }(x) &=&u_{1}(x,\hat{y}(x))-k+\hat{y}^{\prime
}(x)u_{2}(x,\hat{y}(x))-k\hat{y}^{\prime }(x) \\
&=&\hat{y}^{\prime }(x)\left[ u_{2}(x,\hat{y}(x))-k\right] ,
\end{eqnarray*}where the second line follows from the first period first order condition,
equation (\ref{eq:focfirstperiod}). Since the second period consumption is
always distorted, the term in square brackets is always positive, and hence $W$ is increasing in $x$ when $\hat{y}$ is increasing in $x$ (i.e., when utility is supermodular), and
decreasing in $x$ when $\hat{y}$ is decreasing in $x$ (i.e., when utility is supermodular).
\end{proof}

\subsection{An example}\label{sec:numexample}

\begin{table}[tbp]
\begin{center}
\caption{Example}\label{tab:example2period}
\begin{tabular}{l l r r r r r r}
\hline\hline
								& 			& $x$ 	& $y$  	& $W$ 	& $\pi_1$ 	& $\pi_2$ 	& $V$ 	\\\hline 
\multirow{3}{*}{$a=-1$}			& main model& Fig.\ref{fig:a1} 	& Fig.\ref{fig:y1}  	& 6.99 	& 	[2.92, 3.00] 	& 2.92	& [1.08, 1.16] 	\\
								& first best& 1.00 	& 1.00  	& 7.00 	& - 	& - 	& - 	\\ 
								& observable $x$& 1.17 	& 0.97  	& 6.92 	& 4.08	& 2.84 	& 0.00 	\\\hline  
\multirow{3}{*}{$a=-3$} 		& main model& Fig.\ref{fig:a3} 	& Fig.\ref{fig:y3}  	& 5.39 	& [1.36, 1.81] 	& 1.31 	& [2.27, 2.72] 	\\ 
								& first best& 0.78 	& 0.78  	& 5.44 	& - 	& - 	& - 	\\ 
								& observable $x$& 1.17 	& 0.58  	& 5.10 	& 4.08 	& 1.02	& 0.00 	\\ \hline\hline
\end{tabular}
\end{center}
\end{table}

We now consider a numerical example, where $u$ is given by:
\begin{align*}
&u(x,y) = -3x^2-3y^2+axy+8x+8y
\end{align*}
The parameter $a = u_{21}(x,y)$ is a measure of substitutability of the past
and current consumption. We focus on the submodular case, so that $a$ is negative, and focus on two values,  $a=-1$ and $a=-3$.  The equilibrium values for the variables of interest are presented in Table~\ref{tab:example2period}. The equilibrium distributions of first period
consumption for the two cases are given in Figure \ref{fig:a}.

\begin{figure}
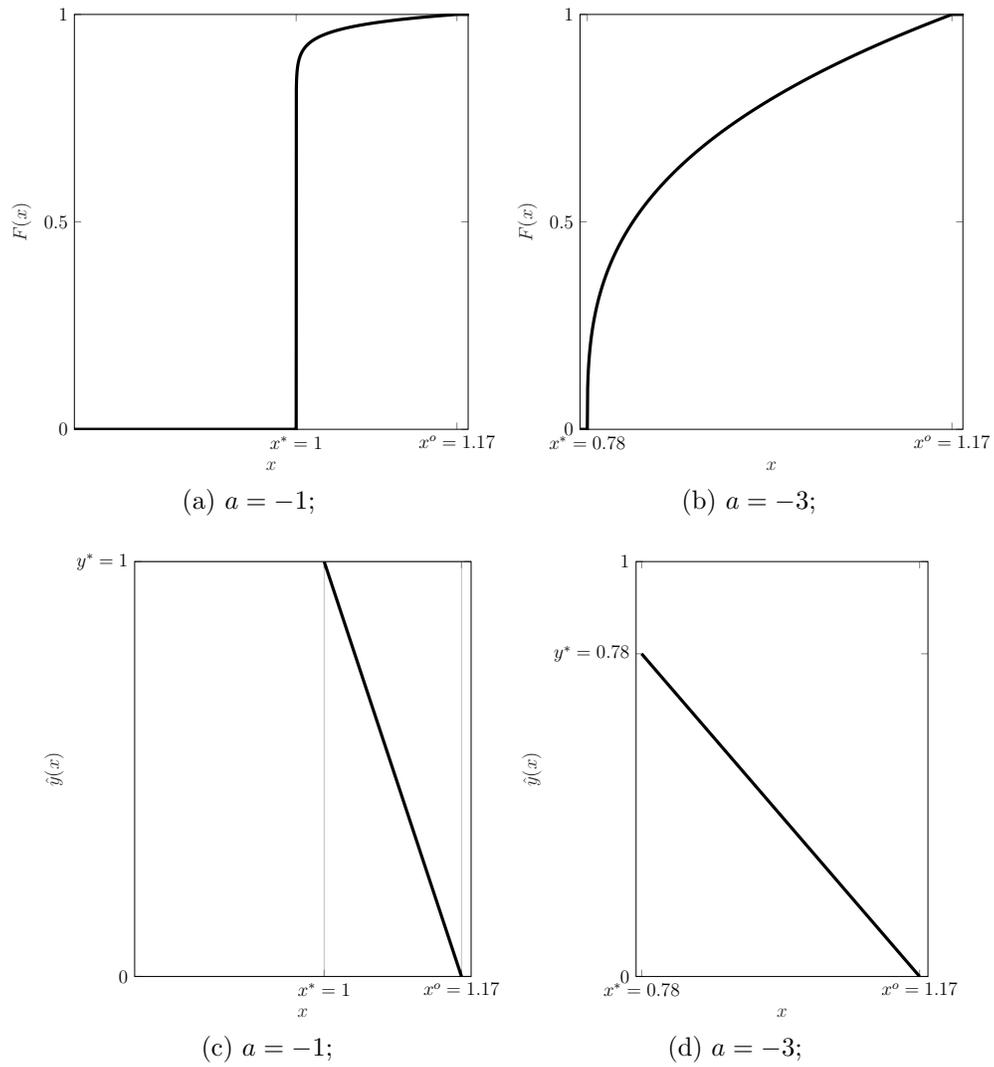

    \centering
    \subfloat[$a=-1$;\label{fig:a1}]{
        \input{fig_distr1.tex}}
	\subfloat[$a=-3$;\label{fig:a3}]{
        \input{fig_distr3.tex}}\\
    \subfloat[$a=-1$;\label{fig:y1}]{
        \input{fig_y1.tex}}
	\subfloat[$a=-3$;\label{fig:y3}]{
        \input{fig_y3.tex}}
\caption{Distribution of first period consumption (top row) and second period consumption (bottom row).}\label{fig:a}
\end{figure}

In these two cases, the equilibrium distribution is skewed to the left: most
of the consumers consume an amount close to the socially efficient one. This
is also reflected in the fact that efficiency loss in equilibrium, $W^{\ast}-W$, is small compared to the one in benchmark model with observable
consumption, $W^{\ast}-W^o$.

The distribution of the social welfare is also interesting: consumers and
the second period seller obtain higher payoffs under privacy than
 when past consumption is observable. The
profit comparison for the first period seller is the opposite of that: this
seller's profit is reduced by privacy.

Note that the first seller's
loss from privacy equals the  consumer's gain and the increase in  social welfare  accrues fully to 
the second seller, as one would expect from  the results in section~\ref{sec:payoffs}. 
The example shows that  the gain from privacy is much larger for the consumer than for 
seller 2.

\section{Concluding remarks}

We study optimal nonlinear pricing when transactions are private, in the presence of intertemporal substitutability and complementarity. We show that, even if consumers are ex ante identical, the equilibrium menus offered by sellers feature a large variety of bundle sizes paired with quantity discounts. These offerings give rise to endogenous taste heterogeneity across consumers.

In equilibrium, the seller who faces consumers without private information offers a two-part tariff to facilitate the creation of endogenous heterogeneity. Subsequently, the seller who serves consumers with endogenous private information offers a screening menu.

We derive testable relationship between the price dispersion and the degree of inter-temporal substitutability or complementarity. In our model, if past choices are not observed by the sellers, the consumers necessarily retain some surplus in the form of information rent. The informational content of the past choices depends on the degree of inter-temporal substitutability or complementarity and it affects the optimal menus offered by the sellers. Our substantive policy finding is that privacy improves welfare and consumer surplus, and mitigates the distortions due to seller monopoly power. However, in our model, firms will not voluntarily protect consumer privacy, even when consumers are sophisticated, since making transactions public is a way for upstream firms to enhance their monopoly power. Consequently, there is  a role for regulation. 

\appendix

\section{Appendix}\label{sec:nopurestr}

\subsection{Proof of Proposition~\ref{prop:nopurestr}}

If the consumption in period 1 is deterministic, say $\tilde{x}$, any best reply by seller 2 will include an item $(q,y)$ in the second period menu such that the consumer chooses this item and 
\begin{align*}
&u(\tilde{x},y)-q = u(\tilde{x},0)\\
&u_{2}(\tilde{x},y) = k.
\end{align*}
If such a menu is expected, seller 1 has a profitable deviation---he can induce a consumption $x>\tilde{x}$ in the first period and receive a higher profit as a result.  Formally, this is established by the following lemma.

\begin{lemma}\label{lemma:deviations}
Suppose the consumer buys $\tilde{x}$ in the first period and chooses one of the two different available options---$(q_1,y_1)$ and $(q_2,y_2)$---from the menu in the second period. If the consumer is indifferent between these two options---i.e., if
\[u(\tilde{x},y_1)-q_1=u(\tilde{x},y_2)-q_2,\]
then one of the sellers has a profitable deviation.
\end{lemma}
\begin{proof}
We prove this lemma for the case of submodular utility $u(x,y)$. The proof for the case of supermodular utility $u(x,y)$ is very similar and therefore omitted. Without loss of generality, let $y_1<y_2$.

First, if $u_{2}(\tilde{x},y_2)< k$, the seller in period 2 can modify his menu to weakly increase his profit. To see that, consider the set 
\[X = \{x\mid u_{2}(x,y(x))< k,\mbox{ and } (x,y(x))\mbox{ is in the support of consumer's strategy}\}
\] 
If the seller replaces every item $y(x), x\in X$ with the item $y^{\ast}(x)$ (recall that $u_{2}(x,y^{\ast}(x))= k$) and reduces the price of this item by $u(x,y(x))-u(x,y^{\ast}(x))$, he will guarantee that a consumer with history $x\in X$ will purchase $y^{\ast}(x)$, and the profit from selling these items will be strictly higher since $u(x,y)-ky$ is maximized at $y^{\ast}(x)$. Therefore, we can restrict our attention to the case of $u_{2}(\tilde{x},y_2)\geq k$.

The difference in profit between the two options that are offered in the second period are
\[
q_2-ky_2-(q_1-ky_1) = u(\tilde{x},y_2)-u(\tilde{x},y_1)-k(y_2-y_1)>0.
\]
If the consumer chooses $(q_1,y_1)$, seller 2 can reduce $q_2$ by arbitrarily small but positive amount and, thus, increase his profit.

If the consumer chooses $(q_2,y_2)$, seller 1 has a profitable deviation. Indeed, if seller 1 induces $x>\tilde{x}$ instead of $\tilde{x}$, the consumer will choose either $(q_1,y_1)$ or another, even smaller bundle.
Similarly, if seller 1 induces $x<\tilde{x}$ instead of $\tilde{x}$, the consumer will choose either $(q_2,y_2)$ or another, even larger bundle.
Thus, the total payoff of the consumer and seller 1, $\Sigma(x),$ is bounded from below by

\begin{equation*}
\Sigma (x)\geq\tilde{\Sigma}(x):=\left\{
\begin{tabular}{l l}
$u(x,y_1)-q_1-kx$ & , if $x\geq \tilde{x}$ \\
$u(x,y_2)-q_2-kx$ & , if $x< \tilde{x}$.\end{tabular}\right.
\end{equation*}

and $\Sigma(\tilde{x}) = \tilde{\Sigma}(\tilde{x}) = u(\tilde{x},y_1)-q_1-k\tilde{x}$. The function $\tilde{\Sigma}(x)$ is continuous and $u$ is submodular, hence
\begin{equation*}
D_{+}\tilde{\Sigma} (\tilde{x})-D_{-}\tilde{\Sigma}(\tilde{x})=u_{1}(\tilde{x},y_1)-u_{1}(\tilde{x},y_2)>0.
\end{equation*}
Therefore $\tilde{\Sigma}(x)$  cannot achieve the maximum at $\tilde{x}$, and so cannot $\Sigma(x)$.
\end{proof}

\subsection{Proof of Lemma \ref{lemma:differentiable}}\label{app:diff}
Fix $x\in X,$ and $\hat{y}(x).$ Consider the payoff of the consumer in the second period, $U(x+\delta )$---this is well defined for $\delta$ sufficiently small
since $U$ is defined on the open interval $I.$ Since $x+\delta $ can choose
the contract chosen by the consumer with the first-period consumption $x,$
\begin{equation*}
U(x+\delta )\geq u(x+\delta,\hat{y}(x))-q(x)
\end{equation*}

Thus, for $\delta >0$

\begin{equation*}
\frac{U(x+\delta )-U(x)}{\delta }\geq \frac{u(x+\delta,\hat{y}(x))-u(x,\hat{y}(x))}{\delta }.
\end{equation*}
The above inequality implies
\begin{equation}
D_{+}U(x):=\lim \inf_{\delta \rightarrow 0+}\frac{U(x+\delta )-U(x)}{\delta }\geq u_{1}(x,\hat{y}(x)).  \label{U'+}
\end{equation}

Since the inequality for $\delta <0$ has a reversed sign, this yields\begin{equation}
D^{-}U(x):=\lim \sup_{\delta \rightarrow 0-}\frac{U(x+\delta )-U(x)}{\delta }\leq u_{1}(x,\hat{y}(x)).  \label{U'-}
\end{equation}

Now, the total payoff of seller 1 and consumer, $\Sigma (x),$ equals
\begin{equation*}
\Sigma (x)=U(x)-kx.
\end{equation*}

Define:
\[D^{+}U(x):=\lim \sup_{\delta \rightarrow 0+}\frac{U(x+\delta )-U(x)}{\delta },
\]
\[
D_{-}U(x):=\lim \inf_{\delta \rightarrow 0-}\frac{U(x+\delta
)-U(x)}{\delta }.
\]

If $x\in X,$ then, since $x$ is chosen, it must maximize $\Sigma (x),$ and
necessary conditions are
\begin{equation*}
\Sigma ^{+}(x)=D^{+}U(x)-k\leq 0,
\end{equation*}

\begin{equation*}
\Sigma ^{-}(x)=D_{-}U(x)-k\geq 0.
\end{equation*}

These inequalities imply $D^{+}U(x)\leq D_{-}U(x).$ In conjunction with the
inequalities (\ref{U'+}) and (\ref{U'-}), this implies that for any $x\in X,$
\begin{equation*}
D^{+}U(x)=D_{+}U(x)=D^{-}U(x)=D_{-}U(x)=u_{1}(x,\hat{y}(x)).
\end{equation*}

\subsection {Proof of Theorem \ref{t:eq}}\label{app:main}

We focus on the case of submodular utility. Since the argument is very similar when utility is supermodular, we omit it.

The proof consists of the following steps:
\begin{enumerate}
\item We derive necessary conditions for profit maximization in periods 1 and 2. We use these conditions to solve for the equilibrium menus and the distribution of the first period consumption $F$. When deriving the optimality conditions we conjecture and later verify that the distribution $F$ has at most one atom (if the atom exists it is at $x^{\ast}$).
\item We show that the solution we get for $F$ is a cumulative distribution function---i.e., it is non-decreasing and right-continuous.
\item We show that the menus that we find induce the conjectured consumer choices (or using the language of the theory of incentives, the menus are incentive-compatible).
\end{enumerate} 

First, we verify that the proposed distribution $F$ induces the second-period consumption $\hat{y}$ and indirect utility $U(x)$.  
The outside option of the consumer in the second period is $u(x,0)$. If the consumer consumes $\overline{x}$ in the first period, her continuation utility in the second period is $u(\overline{x},0)$. Therefore
\begin{equation*}
U(x) = u(\overline{x},0)-\int\limits_{x}^{\overline{x}}u_{1}(z,\hat{y}(z))dz.
\end{equation*}The price charged by the seller 2 for the bundle $\hat{y}(x)$ is
\begin{equation}
q(x)=u(x,\hat{y}(x))-u(\overline{x},0)+\int\limits_{x}^{\overline{x}}u_{1}(z,\hat{y}(z))dz. \label{eq:priceb}
\end{equation}Hence the expected profit for this seller is
\begin{equation*}
\int\limits_{\underline{x}}^{\overline{x}}\left[ u(x,\hat{y}(x))f(x)+u_{1}(x,\hat{y}(x))F(x)-k\hat{y}(x)f(x)\right]dx -u(\overline{x},0).
\end{equation*}Maximizing the above expression pointwise, we obtain that the distribution  $F$ must
satisfy the first order condition
\begin{equation}
u_{2}(x,\hat{y}(x))f(x)+u_{21}(x,\hat{y}(x))F(x)-kf(x)=0.
\label{eq:focsecondperiod}
\end{equation}

Seller 1 makes the consumer indifferent between the inside and the outside options, therefore, he charges a price for amount $x$ that equals

\begin{equation*}
p(x)=u(\underline{x},\hat{y}(\underline{x}))-u(0,\hat{y}(\underline{x}))+\int\limits_{\underline{x}}^{x}u_{1}(z,\hat{y}(z))dz.
\end{equation*}

In particular, the price for the bundle \underline{$x$} equals
\begin{equation*}
p(\underline{x})=u(\underline{x},\hat{y}(\underline{x}))-u(0,\hat{y}(\underline{x})).
\end{equation*}

Seller 1's profit from selling a bundle $x$ has to be independent of $x$.
Hence

\begin{equation}  \label{eq:focfirstperiod}
u_{1}(x,\hat{y}(x))-k = 0.
\end{equation}

Equations (\ref{eq:focfirstperiod}) and (\ref{eq:focsecondperiod}) pin down
unknown functions $F$ and $\hat{y}$.

\begin{lemma}
Suppose that $F$ and $\hat{y}$ solve equations (\ref{eq:focfirstperiod}) and (\ref{eq:focsecondperiod}). Then $F$ is a c.d.f.
and $\hat{y}$ is strictly decreasing (strictly increasing resp.) whenever $u$ is submodular (supermodular resp.).
\end{lemma}

\begin{proof}
By taking a derivative of equation (\ref{eq:focfirstperiod}) with respect to
$x$ we get
\begin{equation*}
\hat{y}^{\prime}(x) = -\frac{u_{11}(x,\hat{y}(x))}{u_{21}(x,\hat{y}(x))}
\end{equation*}
therefore $\sign(\hat{y}^{\prime}(x)) = \sign(u_{21}(x,\hat{y}(x)))$.

The solution to equation (\ref{eq:focsecondperiod}) is
\begin{equation}
\ln F(x)=\int\limits_{x}^{\overline{x}}\frac{u_{21}(z,\hat{y}(z))}{u_{2}(z,\hat{y}(z))-k}dz  \label{eq:focsecondperiodsolution}
\end{equation}This solution is increasing in $x$ if $u_{2}(x,\hat{y}(x))\geq k$ for all $x\in \lbrack \underline{x},\overline{x}]$. Lower bound $\underline{x}$
solves $u_{2}(\underline{x},\hat{y}(\underline{x}))=k$ \citep[see ][]{Hellwig10}. Moreover, since $u$ is concave
\begin{align*}
\frac{d}{dx}u_{2}(x,\hat{y}(x))& =u_{21}(x,\hat{y}(x))+u_{22}(x,\hat{y}(x))\hat{y}^{\prime }(x) \\
& =u_{21}(x,\hat{y}(x))-\frac{u_{11}(x,\hat{y}(x))}{u_{21}(x,\hat{y}(x))}u_{22}(x,\hat{y}(x)) \\
& =\frac{u_{22}(x,\hat{y}(x))u_{11}(x,\hat{y}(x))-(u_{21}(x,\hat{y}(x)))^{2}}{-u_{21}(x,\hat{y}(x))}\geq0.
\end{align*}Therefore, $u_{2}(x,\hat{y}(x))$ is increasing in $x$ and $u_{2}(x,\hat{y}(x))\geq k$ for all $x\in \lbrack \underline{x},\overline{x}]$. To
summarize, the solution $F(x)$ is a non-decreasing right-continuous function and $F(\overline{x})=1$, therefore $F(x)$ is a c.d.f.\footnote{Strictly speaking, the ODE (\ref{eq:focsecondperiod}) has a solution (\ref{eq:focsecondperiodsolution}) on $(\underline{x},\overline{x}]$, but we can
right-continuously extend it to $\underline{x}$ with value of $\lim\limits_{x\to\underline{x}+0}\exp \left[ \int\limits_{x}^{\overline{x}}\frac{u_{21}(z,\hat{y}(z))}{u_{2}(z,\hat{y}(z))-k}dz\right]$. If this value is strictly positive, the distribution $F$ has an atom at $\underline{x}.$}
\end{proof}

The final step in the proof is established by the following lemma, that shows that the consumer will choose option $\hat{y}(x)$ from the second period menu if she consumed $x$ in the first period.

\begin{lemma}
\label{lemma:ic} If $\hat{y}(x)$ is decreasing and utility is submodular (or if $\hat{y}(x)$ is increasing and utility is supermodular), equation \eqref{eq:priceb} implies 
\begin{equation*}
u(x,\hat{y}(t))-q(t)\leq u(x,\hat{y}(x))-q(x)
\end{equation*}for all $x,t\in[\underline{x},\overline{x}]$.
\end{lemma}

\begin{proof}
Consider $x>t$. Since $\hat{y}(x)$ is decreasing and utility is submodular (or,
alternatively, $\hat{y}(x)$ is increasing and utility is supermodular), we have
\begin{align*}
& q(t)-q(x)=u(t,\hat{y}(t))-u(x,\hat{y}(x))+\int\limits_{t}^{x}u_{1}(z,\hat{y}(z))dz\geq \\
& u(t,\hat{y}(t))-u(x,\hat{y}(x))+\int\limits_{t}^{x}u_{1}(z,\hat{y}(t))dz=
\\
& u(t,\hat{y}(t))-u(x,\hat{y}(x))+u(x,\hat{y}(t))-u(t,\hat{y}(t))= \\
& u(x,\hat{y}(t))-u(x,\hat{y}(x)).
\end{align*}The case with $x<t$ is identical.
\end{proof}

\subsection{Proof of Theorem~\ref{t:uniqueness}}\label{app:uniqueness}

Our proof hinges on two facts that have been established. First, for any
equilibrium $\sigma $ with support $X_{\sigma}:U_{\sigma}(x)=U(x),$ and
second, $\hat{y}_{\sigma}(x)$, the bundle consumed in the second period in equilibrium $\sigma$ by the consumer who consumed $x$ in the past, is uniquely determined and coincides with that under $\tilde{\sigma}:\hat{y}(x)$. In other words, the payoff and the second-period consumption for any chosen first-period consumption is the same across all equilibria.

The argument is completed via the following lemma, that shows that any equilibrium distribution cannot have either gaps in the support or atoms (except for a possible atom at $x^{\ast}$). 

\begin{lemma}\label{lemma:nogap}
Let $\sigma$ be an equilibrium, and let $G$ denote the c.d.f. of first period consumption corresponding to $\sigma$.
The support of $G$ equals  $[x^{\ast
},x^{o}]$, and $G$ cannot have atoms. 
\end{lemma}
\begin{proof}
Suppose, by contradiction, that  $G$ has a gap $(x_1,x_2)$ in its support. Sequential rationality for the second-period seller implies that the consumer who consumed $x_1$ in the first period is indifferent between the item that she chooses in the second period and the item that is chosen by the consumer with a history $x_2$. This violates Lemma~\ref{lemma:deviations}. Therefore, distribution $G$ cannot have gaps in its support.

Suppose, by contradiction, that  $G$ has an atom at $\tilde{x}\neq x^{\ast}$. \cite{Hellwig10} establishes that in that case $\hat{y}(x)$---consumption in the second period as a function of the consumer's history---is discontinuous at $\tilde{x}$. Since $G$ has no gaps in its support, by Lemma~\ref{lemma:monotonicity}, $\hat{y}(x)$ must be continuous on $[x^{\ast},x^{o}]$ hence the contradiction.
\end{proof}

\bibliographystyle{aer}
\bibliography{gjm-vb-refs}

\end{document}